\documentclass[11pt]{article}

\usepackage{amsmath}
\usepackage{amsfonts}
\usepackage{amssymb}
\usepackage{amsthm}
\usepackage{enumerate}
\usepackage{bbm}
\usepackage{hyperref}
\usepackage{dsfont}
\usepackage{algorithmicx,algpseudocode,algorithm}
\usepackage{array}
\usepackage{fullpage}
\usepackage{utopia}

\newtheorem{theorem}{Theorem}[section]
\newtheorem{lemma}[theorem]{Lemma}

\newtheorem{claim}[theorem]{Claim}
\newtheorem{proposition}[theorem]{Proposition}

\newcommand{\suchthat}{\;\ifnum\currentgrouptype=16 \middle\fi|\;}

\newcommand{\F}{\mathbb{F}}

\newcommand*{\defeq}{\stackrel{\text{def}}{=}}

\newcommand*{\horzbar}{\rule[.5ex]{2.5ex}{0.5pt}}

\begin{document}

\title{Syndrome decoding of Reed-Muller codes and tensor decomposition over finite fields}
\author{
Swastik Kopparty\thanks{Department of Mathematics\& Department of Computer Science, Rutgers University.
Research supported in part by a Sloan Fellowship, NSF grants CCF-1253886 and CCF-1540634, and BSF grant 2014359.
{\tt swastik.kopparty@gmail.com}.} \and
Aditya Potukuchi\thanks{Department of Computer Science, Rutgers University. {\tt aditya.potukuchi@cs.rutgers.edu}.}}
\maketitle

\begin{abstract}
Reed-Muller codes are some of the oldest and most widely studied error-correcting codes, of interest for both their algebraic structure as well as their many algorithmic properties. A recent beautiful result of Saptharishi, Shpilka and Volk \cite{SSV:15} showed that for binary Reed-Muller codes of length $n$ and distance $d = O(1)$, one can correct $\operatorname{polylog}(n)$ random errors in $\operatorname{poly}(n)$ time (which is well beyond the worst-case error tolerance of $O(1)$).

In this paper, we consider the problem of {\em syndrome decoding} Reed-Muller codes from random errors. More specifically, given the $\operatorname{polylog}(n)$-bit long syndrome vector of a codeword corrupted in $\operatorname{polylog}(n)$ random coordinates, we would like to compute the locations of the codeword corruptions. This problem turns out to be equivalent to a basic question about computing tensor decomposition of random low-rank tensors over finite fields.

Our main result is that syndrome decoding of Reed-Muller codes (and the equivalent tensor decomposition problem) can be solved efficiently, i.e., in $\operatorname{polylog}(n)$ time. We give two algorithms for this problem:

\begin{enumerate}

\item The first algorithm is a finite field variant of a classical algorithm for tensor decomposition over real numbers due to Jennrich. This also gives an alternate proof for the main result of~\cite{SSV:15}.
\item The second algorithm is obtained by implementing the steps of~\cite{SSV:15}'s Berlekamp-Welch-style decoding algorithm in sublinear-time. The main new ingredient is an algorithm for solving certain kinds of systems of polynomial equations.
\end{enumerate}
\end{abstract}

\section{Introduction.}

Reed-Muller codes are some of the oldest and most widely studied error-correcting codes, of interest for both their algebraic structure as well as their many algorithmic properties.
A recent beautiful result of Saptharishi, Shpilka and Volk~\cite{SSV:15} (building on Abbe, Shpilka and Wigderson~\cite{ASW:15}) showed that for binary Reed-Muller codes of length
$n$ and distance $d = O(1)$, there is a $\operatorname{poly}(n)$-time algorithm\footnote{In fact, the algorithm of~\cite{SSV:15} runs in near linear time $n \operatorname{poly}\log(n)$.} that can correct $\operatorname{poly}\log(n)$ random errors (which is well beyond the worst-case error tolerance of $d/2 = O(1)$ errors).
In this paper, we show that the $\operatorname{poly}\log(n)$ random error locations can in fact be computed
in \underline{$\operatorname{poly}\log(n)$ time} given the {\em syndrome vector} of the received word. 
In particular, our main result shows that there is a $\operatorname{poly}(n)$-time, $\operatorname{poly}\log(n)$-space algorithm
that can compute the error-locations.\footnote{This algorithm is in fact a one-pass streaming algorithm
which spends $\operatorname{poly}\log(n)$-time per coordinate as it scans the received word, and at the end of the pass it computes the error-locations
in time  $\operatorname{poly}\log(n)$.}

Syndrome decoding of Reed-Muller codes turns out to be equivalent to a basic problem about tensor decompositions over finite fields.
We give two algorithms for  our main result, one coming from the Reed-Muller code world (and based on~\cite{SSV:15}), and another
coming from the tensor-decomposition world (and based on algorithms for tensor decompositions over the real numbers).

\subsection{Problem setup}
A error-correcting code is simply a subset $C \subseteq \F_2^n$.
We say the code $C$ has minimum distance $\geq d$ if for any distinct $c_1,c_2 \in C$, the Hamming distance $\Delta(c_1,c_2) \geq d$.
The main nontrivial algorithmic task associated with an error-correcting code $C$ is {\em decoding}: for a codeword
$c$ and a sparse error-vector $e$, if we are given the ``received word" $y = c+ e$, we would like to compute the original
codeword $c$.

A {\em linear code} $C$ is a code which is also an $\F_2$-linear subspace of $\F_2^n$.
Let $k$ denote the dimension of the code, and let $k' = n - k$.
Linear codes are usually specified either by giving a {\em generating matrix} $G$ (whose rows span $C$)
or an $k' \times n$ {\em parity-check matrix} $H$ (whose rows span the orthogonal space $C^{\perp}$).
Given a received word $y = c+ e$, where $c$ is a codeword and $e$ is a sparse vector,
the {\em syndrome} of $y$ is simply the vector $S \in \F_2^{k'}$ given by:
$$S = H\cdot y = H\cdot (c+e) = 0 + H\cdot e = H\cdot e.$$
Observe that the syndrome can easily be computed from the received word.
An important fact here is that the syndrome is exclusively a function of $e$, and does not depend on $c$.
Given the syndrome $S= H \cdot y$ (where $y = c+e$ for a codeword $c$ and a sparse error vector $e$), the algorithmic
problem of {\em syndrome decoding} is to compute the error vector $e$. Clearly, a syndrome decoding algorithm
can also be used for standard decoding: given a received word $y$ we can compute the syndrome $H\cdot y$, and
then apply a syndrome decoding algorithm to it.

Reed-Muller codes are algebraic error-correcting codes based on polynomial evaluation \linebreak \cite{REED:54,MULLER:54}.
Here we focus on Reed-Muller codes over $\F_2$ with constant distance (although our results apply to larger fields and larger distances too).
Let $m$ be a large integer, and let $r = O(1)$ be an integer.
Associated to these parameters, the Reed-Muller code $RM(m,m-r)$ is defined as follows.
The coordinates of the code correspond to the points of $\F_2^m$ (and thus the length $n = 2^m$.
To each polynomial $P(X_1, \ldots, X_m)$ of individual degree $\leq 1$ and total degree $\leq m-r$,
we associate a codword in $RM(m,m - r)$: this codeword is given by evaluations of $P$ at all the points
of $\F_2^m$. This code has codimension $\Theta(m^{r}) = \Theta((\log n)^r)$ and
 minimum distance $d = 2^{r} = \Theta(1)$.

Decoding algorithms for Reed-Muller codes have a long history.
It has been known for a long time that one can decode from $d/2$ worst case errors in polynomial time (recall that $d$ is the distance of the code).
There has been much work on decoding these codes under random errors~\cite{DUMER:17} and the local testing,
local decoding and local list-decoding 
of these codes~\cite{BLR:93,RubinfeldSudan:96,GoldreichLevin:89,AroraSudan:03,STV:01,AKKLR:05,BKSSZ:10}.

A recent beautiful and surprising result of Saptharishi, Shpilka and Volk      \cite{SSV:15} (building on Shpilka Abbe, Shpilka and Wigderson~\cite{ASW:15}, Kumar and Pfister~\cite{KP:15}, and Kudekar et.al. \linebreak \cite{KMSU:15}) gave new insights into the error-correction capabilities of Reed-Muller codes under random errors.
In the constant distance regime, their results showed that the above Reed-Muller codes $RM(m, m-r)$ (with codimension $\Theta \left((\log n)^r \right)$ and
distance $O(1)$ can in fact be decoded in $\operatorname{poly}(n)$ time from $\Theta \left((\log n)^{\lfloor(r-1)/2 \rfloor} \right)$ random errors with high probability (which is well beyond the
worst-case error-correction radius of $O(1)$). 

Our main result is a syndrome decoding version of the above.

\medskip \noindent{\bf Theorem A (Informal)\ :}\ \ {\em Let $c \in RM(m, m - r)$ be an arbitrary codeword, and let
$e \in \F_2^n$ be a uniformly random string with Hamming weight at most $o\left((\log n)^{\lfloor (r-1)/2\rfloor}\right)$.
There is a deterministic $(\log n)^{O(r)}$ time algorithm, which when given
the syndrome $S = H\cdot e$, computes the set of nonzero coordinates of $e$ (with high probability over
the choice of $e$).}
\medskip

As an immediate corollary, there is a streaming algorithm for computing the error-locations in the above setting,
which makes one pass over $y$, uses only $\operatorname{poly}\log(n)$ space, and spends only $\operatorname{poly}\log(n)$ time per coordinate.
Indeed, the syndrome $H \cdot y$ (where $H$ is parity check matrix of Reed-Muller codes) can be easily computed 
in one pass over $y$ (using the $\operatorname{poly}\log(n)$ space and $\operatorname{poly}\log(n)$ time per coordinate),
after which the syndrome decoding algorithm of Theorem A can compute the nonzero coordinates of $e$.

\subsection{Techniques}

We give two proofs of our main result. The first goes via a connection
to the problem of tensor-decomposition of random low-rank tensors over finite fields. We give an efficient algorithm for this tensor-decomposition problem, by
adapting a known algorithm (due to Jennrich) for the analogous problem over the real numbers. The second goes via the original approach of~\cite{SSV:15}, which
is a novel variant of the Berlekamp-Welch decoding algorithm for Reed-Solomon codes. We show how to implement their steps in a compact form; an important technical
step in this is a new algorithm to solve certain systems of polynomial equations, using ideas related to the Valiant-Vazirani isolation lemma.

\subsubsection{Approach via tensor decomposition}
It will be useful to understand how a parity-check matrix $H$ of the Reed-Muller code $RM(m, m - 2r - 2)$ looks.
Recall that $H$ is a $k' \times n$ matrix (where $k'$ is the codimension
of $RM(m,m - 2r - 2)$ in $\F_2^n$). The rows of $H$ are indexed by elements of $\F_2^m$, and for $x \in \F_2^m$,
the $x$-column of $H$ turns out to (essentially) equal $x^{\otimes \leq 2r + 1}$, the $\leq 2r+1$'th tensor powers of $x$.
Thus for a random sparse vector $e$ whose nonzero coordinates are $E \subseteq \F_2^m$, the syndrome $S = H \cdot e$ ends up equalling:
$$ S = \sum_{e \in E} e^{\otimes \leq 2r + 1}.$$

Having written the problem in this way, the problem of computing the error locations $E$ from the syndrome $S$ is basically just the problem of
tensor decomposition of an appropriately constructed random low rank tensor over $\F_2$.

We show how this problem can be solved efficiently. We adapt an elegant algorithm of Jennrich for this task over the real numbers.
This algorithm is based on taking two random flattenings of the tensor $S$ into matrices, using properties of the pseudoinverse (a.k.a.
the Moore-Penrose generalized inverse) of a singular matrix, and spectral ideas.
Two difficulties show up over finite fields. The more serious one is that the Moore-Penrose generalized inverse does not exist in general
over finite fields~\cite{ginversesbook} (and even in our special situation). We overcome this by developing an alternate algorithm that does not use the
pseudoinverse of a singular matrix, but instead keeps track of a full rank minor of the singular matrix.
The other difficulty is that small finite fields do not have enough elements in them for a matrix to have all distinct
eigenvalues in the field. We overcome this by moving to a large enough extension field $\F_{2^{10m}}$.

Finally we note that this gives a new proof of the main theorem of~\cite{SSV:15}.
The details appear in Section~\ref{sec:jennrich}. There we also briefly sketch how to derandomize this algorithm.

\subsubsection{Approach via solving polynomial equations}
The original approach of~\cite{SSV:15} works as follows. Given the received word $y \in \F_2^n$,
we view it as a function from $F_2^m \to \F_2$. We then look for all polynomials $A(X_1, \ldots, X_m), B(X_1, \ldots, X_m)$
of degree at most $r+1$, $m-r - 1$ respectively, such that for all $x \in \F_2^m$:
$$A(x) \cdot y(x) = B(x).$$
\cite{SSV:15} suggested to consider the linear space $V$ of {\em all} $A(X_1, \ldots, X_m)$ for which there
exists such a $B(X_1, \ldots, X_m)$\footnote{This idea of considering {\em all} solutions of this
``error-locating equation" instead of just one solution is the radical new twist over the Berlekamp-Welch 
algorithm that makes~\cite{SSV:15} so powerful.}. The main property they show is that 
for $E$ is completely characterized by $V$; namely, $E$ is precisely the set of common zeroes of all the elements of $V$.
Then~\cite{SSV:15} simply check for each point $x \in \F_2^m$ whether it is a common zero of all elements of $V$.

Our syndrome decoder tries to do the same, in $\operatorname{poly}(m)$ time instead of $\operatorname{poly}(2^m)$ time, using only the syndrome.
We begin by observing that a basis for the space $V$ can be found given only the syndrome of $y$.
This reduces us to the problem of finding the common zeroes of the collection of polynomials in $V$.

In full generality, given a collection of low degree polynomial finding their common solutions is NP-hard.
Indeed, this very easily encodes SAT. However our situation is different in a subtle but important way.
It turns out that $V$ is the space of {\em all} low degree polynomials that vanish on $E$. So we are not solving
an arbitrary system of polynomial equations! The next theorem says that such systems of polynomial equations are
solvable efficiently.

\medskip
\noindent{\bf Theorem B (Informal):}\ \ {\em 
Let $E \subseteq \F_2^m$ be a uniformly random subset of size $o(m^r)$.
Let $V$ be the space of all polynomials of degree at most $r+1$ which vanish on $E$.
There is a deterministic polynomial time algorithm that, when given a basis for $V$ as input, computes $E$ (with
high probability over the choice of $E$).
}
\medskip

Our algorithm for this problem uses ideas related to the Valiant-Vazirani isolation lemma (which reduces
SAT to Unique-SAT). If $E$ turned out to be of size exactly $1$, it turns out that there is a very
simple way to read off the element of $E$ from $V$. 
We show how to reduce the general case to this case: by choosing a random affine subspace $G$ of a suitable
small codimension $c$, we can ensure that $|E \cap G| = 1$. It also turns out that when $E$ is random,
given the space of all $m$-variate polynomials of degree at most $r+1$ vanishing on $E$, we can compute the space of all
$m-c$-variate polynomials (viewing $G$ as $\F_2^{m-c}$) of degree at most $r+1$ vanishing on $G \cap E$. This lets
us reduce to the case of a unique solution, and we can recover an element of $E$. Repeating this several
times gives us all elements of $E$.

We also give a different algorithm for Theorem B using similar ideas, which has the advantage of being deterministic. The
key subroutine for this algorithm is that given an affine subspace $H \subseteq \F_2^m$, we can compute the size of 
$E \cap H$ from $V$ (for this it is important that $E$ is random). This subroutine then easily allows us to zoom in on the elements of
$E$.

\section{Notation.}
\begin{itemize}
\item We say that $a = b \pm c$ to mean $a \in [b-c, b+c]$.
\item We use $\omega$ to denote the exponent of matrix multiplication.
\item For a matrix $M_{m \times n}$, and subsets $A \subseteq [m]$ and $B \subseteq [n]$, we say $M_{A,B}$ to mean to submatrix of $M$ with rows and columns indexed by elements in $A$ and $B$ respectively. Further, $M_{A,\cdot} \defeq M_{A,[n]}$, and $M_{\cdot, B} \defeq M_{[m],B}$.
\item We use $\mathcal{M}_r^n$ to denote the set of all monomials of degree $\leq r$ in $n$ variables $X_1,\ldots, X_n$.
\item For a vector $v \in \F_2^m$, let us write $v^{\otimes \leq t}$ to mean the vector of length $\binom{m}{\leq t}$, whose entries are indexed by the monomials in $\mathcal{M}_t^m$. The entry corresponding to $M \in \mathcal{M}_r^m$ is given by $M(v)$.
\item For a set of points $A \subseteq \F_2^n$, we use $A^{\otimes \leq t} \defeq \{v^{\otimes \leq t} \suchthat v \in A\}$.
\item A set of points $A \subseteq \F_2^n$ is said to satisfy property $U_r$ if the vectors in $A^{\otimes \leq r}$ are linearly independent.
\item For a set $A\subseteq \F_2^t$, we denote $\operatorname{mat}(A)$ to be the $|A| \times t$ matrix whose rows are elements of $A$.
\end{itemize}

\section{The Main Result.}

The main result is that we show how to decode high-rate Reed Muller codes $RM(m, m-2r - 2)$, where we think of $r$ as growing very slowly compared to $m$, say, a constant. In this case, the received corrupted codeword is of length $n = 2^m$. However, \emph{syndrome} of this code word is $O(m^{2r})$. We want to find the set of error locations from the syndrome itself \emph{efficiently}. Formally, we prove the following:

\begin{theorem}
\label{thm:main}
Let $E$ be a set of points in $\F_2^m$ that satisfy property $U_r$. There is a randomized algorithm that takes as input, the syndrome of an $RM(m, m - 2r - 2)$ codeword corrupted at points in $E$, and returns the list of error locations $E$ with probability $> .99$. This algorithm runs in time $O(m^{\omega r+4})$.
\end{theorem}

Our first proof of this theorem is via the `Tensor Decomposition Problem' over small finite fields. As the name suggests, this is just the finite field analogue of the well-studied Tensor Decomposition problem (see, for example, \cite{TENSOR}). The problem is (equivalently) stated as follows: Vectors $e_1, \ldots, e_t$ are picked uniformly and independently from $\F_2^m$. We are given access to
\[
\sum_{i \in [t]}e_i^{\otimes \leq 2r + 1},
\]
and the goal is to recover $e_i$'s. The fact that the $e_i$'s are picked randomly is extremely important, as otherwise, the $e_i$'s can be picked so that the decomposition is not unique. We rely on the results from \cite{ASW:15}, \cite{KMSU:15} and \cite{KP:15}, which informally state that the Reed-Muller codes achieve capacity in the Binary Erasure Channel (BEC) in the very high rate regime, entire constant rate regime, and the very low rate regime. More precisely, for $RM(m,d)$ when the degree $d$ of the polynomials is $o(m)$, $m/2 \pm O(\sqrt{m})$, $m - o(\sqrt{m/\log m})$. This means that when a set of points are picked independently with probability $p$, where $p = 1 - R - \epsilon$, where $R$ is the rate of the code, and $\epsilon$ is a small constant, these points satisfy property $U_r$ with high probability for this range of $R$.

Since this is a tensor decomposition problem, one natural approach is to try and adapt existing tensor decomposition algorithms. Assuming only that the $e_i^{\otimes \leq r}$'s are linearly independent, we show how to decompose $\sum_{i \in [t]}e_i^{\otimes \leq 2r + 1}$. Indeed, this is a very well studied problem in the machine learning community, and one can adapt existing techniques with a bit of extra work. The advantage of this approach is the simplicity and its ability to give the proof the main result of \cite{SSV:15} by giving an efficient algorithm. 

Our second approach to solving this problem in finite fields is to reduce it to finding the common zeroes of a space of low degree polynomials, which we then proceed to solve. This algorithm goes via an interesting and natural algebraic route involving solving systems of polynomial equations. The running time of the resulting algorithm has a worse dependence on the field size that the first approach.
We note here that this alsogives a new approach to tensor decomposition, using ideas related to the Berlekamp-Welch algorithm.

\section{Approach using Jennrich's Algorithm.}
\label{sec:jennrich}

The key idea is that, we will look the vector $v^{\otimes \leq 2r+1}$ as a $3$-tensor $v^{\otimes \leq r} \otimes v^{\otimes \leq r} \otimes v^{\otimes \leq 1}$. Indeed, given the syndrome $\sum_{i \in [t]}e_i^{\otimes \leq 2r+1}$, one can easily construct the $3$-dimensional tensor $\sum_{i \in [t]}e_i^{\otimes \leq r} \otimes e_i^{\otimes \leq r} \otimes e_i^{\otimes \leq 1}$, so we may assume that we are given the tensor. This allows to use techniques inspired by existing tensor decomposition algorithms \cite{HAR:70,LRA:93} like Jennrich's Algorithm (see \cite{BLUM}). To our best knowledge, this problem has not been previously studied over finite fields. Although we state the result for codes over $\F_2$, it is not hard to see that the proof works almost verbatim over larger fields.

\subsection{An overview and analysis of the algorithm.}

We first restate the problem:

\begin{itemize}
\item[] \emph{Input:} For a set of vectors $E = \{e_1,\ldots, e_t\} \subset \F_2^m$ that satisfy property $U_r$, we are given the syndrome as a $3$-tensor
\[
S = \sum_{i \in [t]}e_i^{\leq r} \otimes e_i^{\otimes \leq r} \otimes e_i^{\leq 1}.
\]
\item[] \emph{Output:} Recover the $e_i$'s
\end{itemize}

Following in the footsteps of Jennrich's Algorithm, we pick random points in $a$ and $b$ in $\F_{2^{10m}}^{m+1}$ and compute the matrices

\[
S^a \defeq \sum_{i \in [t]}\langle a,e_i^{\otimes \leq 1} \rangle e_i^{\otimes \leq r}\otimes e_i^{\otimes \leq r},
\]

and

\[
S^b \defeq \sum_{i \in [t]}\langle b,e_i^{\otimes \leq 1} \rangle e_i^{\otimes \leq r}\otimes e_i^{\otimes \leq r}.
\]

	Computing these matrices is the same as taking the weighted linear combination of the slices of the tensor $T$ along one of its axes. Define the $t \times \binom{m}{\leq r}$ matrix $X \defeq \operatorname{mat}(E^{\otimes\leq r})^T$, so we have the matrices $S^a = XAX^T$, and $S^b = XBX^T$ for diagonal matrices $A$ and $B$ respectively. The $i$'th diagonal entry of $A$ is given by $a_i \defeq \langle a ,e_i^{\otimes \leq 1} \rangle$, and the $i$'th diagonal entry of $B$ is given by $\langle b, e_i^{\otimes \leq 1} \rangle$. Let $K$ and $L$ be two (not necessarily distinct) maximal linearly independent sets of $t (= |E|)$ rows in $X$. Denote $X_{K} \defeq X_{K,\cdot}$, and $X_{L} \defeq X_{L,\cdot}$ as shorthand. We have that $S^a_{K,L} \defeq X_{K}AX_{L}^T$, and $S^b_{K,L} \defeq X_{K}BX_{L}^T$ are full rank, since the diagonal entries of $A$ and $B$ are all distinct and nonzero. Therefore, we have the inverse $(S^b_{K,L})^{-1} = (X_{L}^T)^{-1}B^{-1}X_{K}^{-1}$. Multiplying with $S^a_{L}$, we have:

\begin{align*}
S^a_{K,L}(S^b_{K,L})^{-1} & = X_{K}AX_{L}^T(X_{L}^T)^{-1}B^{-1}X_{K}^{-1} \\
& = X_{K}(AB^{-1})X_{K}^{-1}.
\end{align*}

In order to carry out the operations over an extension field, we need to pick an irreducible polynomial of appropriate degree over $\F_2$. Fortunately, this can also be done in time $\operatorname{poly}(m)$. The reason that $a$, and $b$ are chosen from a large extension field is that it ensures that all the entries of $AB^{-1}$ are also nonzero and distinct w.h.p. So, the columns of $X_K$ are just the eigenvectors of this matrix, which we will then proceed to compute. In order to compute the eigenvalues, we need to factor the characteristic polynomial. Here one can use Berlekamp's factoring algorithm \cite{BERLEKAMP:67}. So, we require the following two lemmas:

\begin{lemma}
\label{lem:diagonal}
For nonzero and distinct $x_1,\ldots,x_t \in \F_2^m$, and a uniformly chosen $a$ and $b$ from $\F_{2^{10m}}^{m+1}$, denote $a_i \defeq \{ \langle a, x_i^{\otimes \leq 1} \rangle$, and $b_i \defeq \langle a, x_i^{\otimes \leq 1} \rangle \big \}$. Then we have that w.h.p, 
\begin{itemize}
\item[$(1)$] $a_1,\ldots,a_t,b_1,\ldots, b_t$ are all distinct and nonzero.
\item[$(2)$] $\not \exists i,j \in [t]$ such that $i \neq j$ and $a_ib_i^{-1} = a_jb_j^{-1}$.
\end{itemize}
\end{lemma}

\begin{proof}
For the proof of $(1)$, we just need to say that there is no subset $S \subseteq [n]$ such that $\langle a, \mathds{1}_S \rangle = 0$, or equivalently, there are no nontrivial linearly dependencies in the entries of $a$. Since we picked $a$ and $b$ from a vector space over a large enough field, there are at most $2^{2(m+1)}$ possibilities for nontrivial linear dependencies, and each occurs with probability $\frac{1}{2^{10m}}$. Therefore, there are nontrivial linear dependencies with probability at most $2^{-7m}$.

To prove $(2)$, first fix $i$ and $j$. W.L.O.G, let $k$ be a coordinate where $x_i^{\otimes \leq 1}$ is $1$ and $x_j^{\otimes \leq 1}$ is zero. Fixing $a$, and all but the $k$'th coordinate of $b$ we see that there is exactly one $b[k]$ such that $a_ib_i^{-1} = a_jb_j^{-1}$. Therefore, with $a$, and $b$ picked uniformly, this equation is satisfied with probability at most $\frac{1}{2^{10m}}$. Therefore, there are $i$ and $j$ that satisfy this equality with probability at most $\frac{m^2}{2^{10m}}$.

Therefore, by union bound, both $(1)$, and $(2)$ are satisfied with probability at least $1 - 2^{-6m}$.

\end{proof}

\begin{lemma}
With $X$, $A$, $B$ as given above, the only eigenvectors of $X_KAB^{-1}X_K^{-1}$ are the columns of $X_K$ with probability at least $1 - 2^{-7m}$.
\end{lemma}

\begin{proof}
Indeed, the columns of $X_K$ are eigenvectors of $X_KAB^{-1}X_K^{-1}$ since it is easy to verify that $(X_KAB^{-1}X_K^{-1})X_K = X_K(AB^{-1})$. Moreover, by Lemma~\ref{lem:diagonal}, with probability at least $1 - 2^{-7n}$, the matrix $AB^{-1}$ has distinct nonzero diagonal entries. Therefore, all the eigenvalues are distinct, and so no other vector in the span of the columns if $X_K$ is an eigenvector.

\end{proof}

\textbf{Remark:} In the traditional Jennrich's Algorithm, after defining the matrices $S^a$, and $S^b$, one usually works with the \emph{pseudo-inverse} of $S^b$. It turns out that a necessary condition for the pseudo inverse to exist is that $\operatorname{rank}(S^b) = \operatorname{rank}(S^b(S^b)^T) = \operatorname{rank}((S^b)^TS^b)$. However, this need not be the case for us, in fact, one can have that $X$ is full rank, yet, $X^TX = 0$, which gives $S^b(S^b)^T = 0$, while $S_b$ still has full rank. Hence, we needed to find a full rank square submatrix of $X$ and use it to determine the rest of $X$. 

To recover the rest of $X$, we make use of the entries $(S_{l,i,1})_{l \in L}$ of $S$. We may assume that the entries of $e_i^{\otimes \leq 1}$ are labelled such that $e_i^{\otimes \leq 1}[1] = 1$, and so $S_{\cdot,\cdot,1} = XX^T$. Now suppose we want to recover row $i$, we set up a system of linear equation in the variables $x[i] = (x[i,1], \ldots, x[i,t])$:

\[
X_L\cdot x[i]^T = S_{i,L,1}^T.
\]

This can be solved since $X_L$ is full rank, and therefore, is invertible.

\subsection{The algorithm and running time.}

Given the analysis above, we can now state the algorithm:

\noindent
\begingroup
\small
\begin{algorithm}
\label{alg:ffjennrich}
\begin{algorithmic}
\\
\Procedure{JennrichFF}{$S$} 
\State $a,b \sim \F_{2^{10m}}^{m+1}$
\State $S^a \gets \sum_{i \in [n+1]}S_{\cdot,\cdot,i}a[i]$
\State $S^b \gets \sum_{i \in [n+1]}S_{\cdot,\cdot,i}b[i]$
\State $K,L \gets$ indices of the largest full rank submatrix of $S^a$
\State $v_1,\ldots, v_t \gets$ eigenvectors of $S^a(S^b)^{-1}$
\State $X_{K} \gets (v_1,\ldots, v_t)$
\State \textbf{initialize} matrix $X$
\For{$1 \leq i \leq {m \choose \leq r} $}
\State $X_{i,\cdot}^T \gets X_{K}^{-1}S_{K,i,1}^T$
\EndFor
\State \textbf{return} $X$
\EndProcedure
\end{algorithmic}
\end{algorithm}
\endgroup

There are several steps in this algorithm. We will state the running time of each step
\begin{itemize}
\item[0.] Constructing the $3$-dimensional tensor from the syndrome takes time $O(m^{2r+3})$
\item[1.] Finding an irreducible polynomial of degree $10m$ takes time $O(m^4)$ (see, for example, \cite{SHOUP94}). This is for constructing $\F_{2^{10m}}^{m+1}$.
\item[2.] Constructing $S^a$, and $S^b$ takes time $O\left( \binom{m}{\leq 2r+1}\right)$. 
\item[3.] Computing $K$,$L$ takes time $O\left( \binom{m}{\leq r}^{\omega} \right)$. 
\item[4.] Inverting $X_K$ takes time at most $O\left(t^{\omega} \right)$. 
\item[5.] Recovering $X$ takes time $O\left(t^2 \binom{m}{\leq r} \right)$. In fact, recovering just the relevant coordinates of $X$ takes time just $O(t^2m)$. 
\item[6.] Factoring degree $\binom{m}{\leq r}$ polynomials over $\F_{2^{10m}}$ takes $O(m^{r+4})$ time. This is required for computing eigenvectors.
\end{itemize}

Therefore, the whole algorithm runs in time $O(m^{\omega r +1})$, compared to the input size of $O(m^{2r})$.

\subsection{A note on derandomization.}

We needed to pick $a$ and $b$ at random to ensure that all the $a_i$'s and $b_i$'s  satisfy the conditions given in Lemma~\ref{lem:diagonal}. In order to ensure this deterministically, set the vectors of polynomials 
\begin{align*}
a(\alpha) & = (1,\alpha, \alpha^2,\ldots, \alpha^m) \\
b(\alpha) & = (\alpha^{3m},\alpha^{3m+2},\ldots, \alpha^{5m})
\end{align*}

For any $x_i, x_j \in \F_2^{m + 1}\setminus \{0\}$, where $x_i \neq x_j$, it is easy to see that the polynomials 

\begin{flalign*}
\langle a(\alpha),x_i \rangle, \\
\langle b(\alpha), x_i \rangle, \\
\langle a(\alpha),x_i \rangle - \langle a(\alpha),x_j \rangle,  \\
\langle a(\alpha),x_i \rangle - \langle b(\alpha),x_j \rangle,  \\
\langle b(\alpha),x_i \rangle - \langle b(\alpha),x_j \rangle, & \text{\hspace{5mm} and} \\
\langle a(\alpha),x_i \rangle\langle b(\alpha),x_j \rangle - \langle a(\alpha),x_j \rangle\langle b(\alpha),x_i \rangle \\
\end{flalign*}

are all nonzero polynomials (see Claim~\ref{claim:lastpoly}) of degree at most $6m$ in $\alpha$. Taking $\alpha$ to be some primitive element of the field $\F_{2^{10m}}$ ensures that is it not a root of any of the above polynomials. Such an element can also be efficiently and deterministically found (see~\cite{SHOUP94}).

\begin{claim}
\label{claim:lastpoly}
For two distinct nonzero elements $x_i. x_j \in \F_2^{m+1}$, the polynomial 
\[
P(\alpha) = \langle a(\alpha),x_i \rangle\langle b(\alpha),x_j \rangle - \langle a(\alpha),x_j \rangle\langle b(\alpha),x_i \rangle
\]
is nonzero.
\end{claim}

\begin{proof}
W.L.O.G, let $x_i > x_j$ lexicographically. Let $u$ be the largest index such that $x_i[u] = 1$, and $x_j[u] = 0$, and let $v$ be the largest index such that $x_j[v] = 1$, so $x_i$ and $x_j$ agree on all coordinates indexed higher than $u$. We claim that monomial $\alpha^{3m + 2u + v - 3}$ in $P(\alpha)$ survives, and therefore $P$ is nonzero. Indeed, this is true since it is easy to see that if this monomial has to be cancelled out, it has to be equal to some $\alpha^{3m+ 2u' + v' - 3}$, where $x_i[u'] = x_j[v'] = 1$ and where $u < u',v' < v$. But by our assumption, $x_i$ and $x_j$ agree on every coordinate indexed higher than $u$, and therefore, $x_i[v'] = x_j[u'] = 1$, and therefore, the monomial $\alpha^{3m+ 2u' + v' - 3}$ is computed an even number of additional times.

\end{proof}

\section{Approach via reduction to common zeroes of a space of polynomials.}

In this section, we prove Theorem~\ref{thm:main} via finding common roots to a space of low degree polynomials. There are two components to this, the first is a reduction to an algebraic problem:

\begin{theorem}
\label{thm:reduction}
Let $y$ be a corrupted codeword from $RM(m, m - 2r - 2)$, with error locations at $E \subseteq \F_2^m$. There is an algorithm, \textsc{SpaceRoots}, that runs in time $O(m^{(r+1)\omega})$ that takes the syndrome of $y$ as input, and returns the space of all reduced polynomials of degree $\leq r+1$ that vanish on $E$. 
\end{theorem}

We are now left with the following neat problem, which is interesting in its own right, namely, finding the roots of a space of low degree polynomials:

\begin{theorem}
\label{thm:findroots}
For a set of points $E \subseteq \F_2^m$ that satisfy property $U_r$, given the space $V$ of all reduced polynomials of degree $\leq r+1$ that vanish on $E$, there is an algorithm,\textsc{FindRoots}, that runs in time $m^{2r}$ , and returns the set $E$ with probability $1 - o(1)$.
\end{theorem}

The rest of this section is will be dedicated to proving Theorem~\ref{thm:reduction}, and setting up the stage for Theorem~\ref{thm:findroots}.

We set up more notation that will be continue to be used in the paper: For a vector $v \in \F_2^{2^m}$, we will treat $v$ as a function from $\F_2^m$ to $\F_2$ and vice versa in the natural way, i.e., for a point $x \in \F_2^m$, $v(x)$ is the coordinate corresponding to point $x$ in $v$.

We shall use the following theorem from \cite{SSV:15} that completely characterizes the space of polynomials $V$ that we are looking for:

\begin{theorem}
\label{thm:ssv}
For a set of points $E$ satisfying property $U_r$, let $y$ be the codeword from $RM(m,m - 2r - 2)$ which is flipped at points in $E$. Then, there exists nontrivial polynomials $A$, and $B$ that satisfy:
\begin{align*}
A(x)\cdot y (x) = B(x) & \text{ $\forall x \in \F_2^m$}, 
\end{align*}

where $deg(A) \leq r+1$ and $deg(B) \leq m - r - 1$. Moreover, $E$ is the set of common zeroes of all such $A$'s which satisfy the equations. Furthermore, for 
every polynomial $A$ that vanishes on all points of $E$, there is a $B$
such that the above equation is satisfied.
\end{theorem}

So, the way we prove Theorem~\ref{thm:reduction} is by finding polynomials $A(X_1,\ldots, X_m)$ of degree $\leq r+1$ such that $A\cdot y$ is a polynomial of degree $\leq m - r - 1$. Most important, we would like to find this space $V$ of polynomials \emph{efficiently}, i.e., in time $\operatorname{poly}(m^{2r})$. For this, we set up a system of linear equations and solve for $A$.

\begin{proof}[Proof of Theorem~\ref{thm:reduction}]

Let us use $s$ to denote the syndrome vector, whose entries are indexed by monomials of degree at most $2r + 1$. Let us denote 

\[
A = \sum_{M \in \mathcal{M}_{r+1}^m}a_M M(X_1,\ldots, X_m)
\]

to be a polynomial whose coefficients are indeterminates $a_M$ for $M \in \mathcal{M}_r^m$. We want that $A\cdot y$ is a polynomial of degree $\leq m - r - 1$, so we look at it as a codeword of $RM(m, m-r-1)$. Using the fact that the dual code of $RM(m, m-r-1)$ is $RM(m, r)$, we have, for any monomial $M' \in \mathcal{M}_r^m$,

\begin{align*}
0 & = \sum_{x \in \F_2^m}A(x)y(x)M'(x) \\
& = \sum_{M \in \mathcal{M}_{r+1}^m}a_M\sum_{x \in \F_2^n}y(x) M(x)M'(x) \\
& = \sum_{M \in \mathcal{M}_{r+1}^m}a_M s_{M\cdot M'}.
\end{align*}

Hence, the solution space to  this system of $|\mathcal{M}_r^m|$ equations in $|\mathcal{M}_{r+1}^m|$ variables gives us the space $V$ of all polynomials of degree $\leq r+1$ that vanish on all points of $E$. Moreover, we can do this efficiently in time $O(m^{(2r + 1)\omega})$ using gaussian elimination.

\end{proof}

 Once we have the above result, we can give a proof of Theorem~\ref{thm:main} assuming Theorem~\ref{thm:findroots}.

\begin{algorithm}
\label{alg:polyspace}
\begin{algorithmic}
\\
\Procedure{SyndromeDecode}{$S$}
\State $V \gets$\textsc{ SpaceRoots}($S$)
\State \textbf{return} \textsc{FindRoots}($V$)
\EndProcedure
\end{algorithmic}
\end{algorithm}

Of course, assuming we have an algorithm such as \textsc{FindRoots}, it is obvious that the above algorithm is, indeed what we are looking for. Most of the rest of the paper goes into finding such an algorithm.

\section{Efficiently finding roots of a space of polynomials.}

At this point, we are left with the following neat problem:

\begin{itemize}
\item[] \emph{Input:} Given the space $V$ of all polynomials in $m$ variables degree $r+1$ polynomials which vanish on a set of points $E$ satisfying property $U_r$.
\item[] \emph{Output:} The set $E$.
\end{itemize}

\subsection{A sketch of the rest of the algorithm.}

Let us denote $t = |E|$. The main idea in the rest of the algorithm is to restrict the set of points to only those lying on a randomly chosen affine subspace of codimension $\sim \log t$. The hope is that exactly one point in $E$ lies in this subspace. This happens with constant probability, and in fact, for every point $e \in E$, $e$ is the only point that lies in this subspace with probability at least $\frac{1}{4t}$. This is given by the Valiant-Vazirani lemma. If we could somehow find all multilinear polynomials of degree $\leq r+1$ on this subspace that vanish at $e$, we can just recover this point with relative ease. 

We repeat the above procedure $O(t \log t)$ times, and we will have found every error point with high probability.

In order to get into slightly more detail, we will set up the following notation, which we will continue to use:

\begin{itemize}
\item For $i \in [m]$, let us denote $E_i$ to be the set of errors left after restricting the last $i$ variables to zero and dropping these coordinates, i.e., $E_i \defeq \{x \suchthat  (x_1,\ldots, x_{m-i},0,\ldots, 0) \in E \}$.
\item For $i \in [m]$, let us denote $V_i \subseteq \F_2[X_1,\ldots, X_{m-i}]$ be the space of all polynomials of degree $\leq r+1$ vanishing on $E_i$.
\end{itemize}

Here is another way to look at the above approach which makes the analysis fairly straightforward: Suppose in the (initially unknown) set $E$, we restricted ourselves only to points that lie on $X_m = 0$, and the number of points is strictly less than $|E|$. We can find the space of all degree $\leq r+1$ polynomials $V_1$ that vanish on this subset by simply setting $X_m = 0$ in all the polynomials in $V$ (see Section~\ref{sec:restrict}). Thus, we have reduced it to a problem in $\F_2^{m-1}$ with fewer points.

Suppose after setting the last $k$ variables to zero, there is exactly one point $e$ left. Let the space of polynomials that vanish on this point be $V_k$. We observe that for every $i \in [m-k]$, there is a polynomial $a - X_i \in V_k$  for exactly one value of $a \in \F_2$. This is because $V_k$ is the space of \emph{all} polynomials of degree $\leq r+1$ that vanish on $e$. So $e(i) - X_i$, for every $i \in [m-k]$, is a degree $1$ polynomial vanishing on $e$, and therefore must belong to $V^k$. In fact, we can `read off' the first $m-k$ coordinates of $e$ from $V_k$ by looking at these polynomials. The other coordinates, as dictated by our restriction, are $0$ (see Section~\ref{sec:FindUniqueRoot}).

Of course, there are a few problems with the above approach. Firstly, restricting $E$ to $X_n = 0$ might not reduce the size at all. This is exactly where the randomness of the invertible affine transformation comes to use. The idea is that this ensures that around half the points are eliminated at each restriction. For appropriately chosen $k$, the Valiant-Vazirani Lemma (see Section~\ref{sec:valiantvazirani}) says that an affine linear restriction of codimension $k$ isolates exactly one error location with constant probability. Next, a subtle, but crucial point is that after an invertible linear transformation, the set of points $E$ must still satisfy property $U_r$. Fortunately, this is not very difficult either (see Section~\ref{sec:propertyur}).

A final remark is that we store the error location once we find it, but thinking back to the decoding problem, once an error is found, it can also be directly corrected. This step is easy. For example, over $\F_2$, an error location $e$ is corrected by adding the vector $e^{\otimes 2r + 1}$ to the syndrome $S$. Over fields of other characteristics, adding $e^{\otimes 2r+1}$ to the syndrome does not ensure that the error at location $e$ has been corrected. However, this isn't a problem because any error location that has not been corrected will be found again. At this point, we add a different multiple of $e^{\otimes 2r+1}$ to the syndrome and continue.   

We will now proceed to analyze each of the above mentioned steps separately.

\subsection{Counting the number of error locations.}
\label{sec:FindUniqueRoot}

We have briefly mentioned that we can check if the size of the set $E_i$ at stage $i$ is $1$ or not. However, something more general is true: we can \emph{count} the number of error locations left $|E_i|$ at any stage. This more general fact will prove to be especially useful in the derandomization of this algorithm, given in Section~\ref{sec:derandomization}. It basically follows from the following simple fact:

\begin{claim}
\label{lem:restrank}
The vectors in  $E_i^{\otimes \leq r}$ are linearly independent.
\end{claim}

\begin{proof}
Consider vector $e^{\otimes \leq r} \in E_i^{\otimes \leq r}$. The entries of $e$ come from a fixed subset of the nonzero coordinates of some vector $\tilde{e}^{\otimes \leq r} \in E^{\otimes \leq r}$. Since, the vectors in $E^{\otimes \leq r}$ are linearly independent, it follows that the vectors in $E_i^{\otimes \leq r}$ are also linearly independent.

\end{proof}

Let $E_i = \{e_1,\ldots, e_k\}$. Since $e_1^{\otimes \leq r}, \ldots, e_k^{\otimes \leq r}$ are linearly independent, we have that \linebreak $e_1^{\otimes \leq r+1}, \ldots, e_k^{\otimes \leq r+1}$ are also linearly independent. Therefore, $V_i$ given by the null space of the matrix $\operatorname{mat}(E_i^{\otimes \leq r+1})$, which (recall) is given by:

\[
  \left(
  \begin{array}{ccc} 
    \horzbar & e_1^{\otimes \leq r+1} & \horzbar \\
     & \vdots &  \\
    \horzbar & e_k^{\otimes \leq r+1} & \horzbar
  \end{array}
  \right) \quad
\]

and as a consequence, has codimension exactly equal to the number of points in $E_i$. This gives a general way to count the number of error points that we are dealing with. Thus we have:

\begin{equation}
\label{eqn:number}
|E_i| = \operatorname{codim}(V_i).
\end{equation}

However, in the case where there is just one point, $e$, it is easier to check, and even recover the point. The idea is that for every $j \in [m]$, exactly one of $X_j$ and $1 - X_j$ is in $V$ depending on whether $e(j) = 0$ or $e(j) = 1$ respectively. In this spirit, we define the algorithm to read off a point given the space of all degree $\leq r+1$ polynomials vanishing on it, in fact, the algorithm returns $\bot$ if there isn't exactly one point.

\begin{algorithm}
\label{alg:FindUniqueRoot}
\begin{algorithmic}
\\
\Procedure{FindUniqueRoot}{$V$}
\State \textbf{Initialize} $e$ 
\For{$j \in [m]$}
\If{$X_j \in V \And 1 - X_j \not \in V$}
\State $e(j) \gets 0$
\ElsIf{$1 - X_j \in V \And X_j \not \in V$}
\State $e(j) \gets 1$
\Else
\State \textbf{return} $\bot$
\EndIf
\EndFor
\State \textbf{return} $e$
\EndProcedure
\end{algorithmic}
\end{algorithm}

\subsection{Applying a random invertible affine map.}
\label{sec:propertyur}
We had also briefly mentioned that when we analyze the algorithm, we are applying a random invertible affine map to $\F_2^m$. We do need to prove that after this map, the set of points still satisfy property $U_r$.

\begin{proposition}
For an invertible affine map $L$, if $E$ satisfies property $U_r$, then $L(E)$ also satisfies $U_r$.
\end{proposition}

\begin{proof}
There is a bijection between set of degree $\leq r$ reduced polynomials vanishing on $E$, and the set of degree $\leq r$ reduced polynomials vanishing on $L(E)$ given by applying the map $L$ to variables. For a reduced polynomial $P(X_1,\ldots, X_m)$ of degree $\leq r$ vanishing on $E$, we have that the polynomial $\operatorname{reduce}(P(L^{-1}(X_1), \ldots, L^{-1}(X_m)))$  vanishes on $T(E)$. Moreover, this is unique, in that no other $P'(X_1,\ldots, X_m)$ maps to this polynomial. This is easy to see since there is a unique way to go between the evaluation tables of $P(X_1,\ldots, X_m)$ and $\operatorname{reduce}(P(L^{-1}(X_1), \ldots, L^{-1}(X_m)))$, and no two distinct reduced polynomials have the same evaluation tables.

Similarly, for a reduced polynomial $Q(X_1,\ldots, X_m)$ of degree $\leq r+1$ vanishing on $T=L(E)$, we have that the polynomial $\operatorname{reduce}(Q(L(X_1), \ldots, L(X_m)))$ vanishes on $E$, and no other $Q'(X_1,\ldots, X_m)$ maps to this polynomial. Therefore, the number of points in the null space of $\operatorname{mat}(E^{\otimes \leq r})^T$ has the same size as the null space of $\operatorname{mat}(L(E)^{\otimes \leq r})^T$. Therefore, the spaces has the same codimension, and the rank of $\operatorname{mat}(L(E)^{\otimes \leq r})$ is the same as the rank of $\operatorname{mat}(L(E)^{\otimes \leq r})$.

\end{proof}

\subsection{The Valiant-Vazirani isolation lemma.}
\label{sec:valiantvazirani}

Here, we will make use of a simple fact, commonly referred to as the Valiant-Vazirani Lemma \cite{VV:86} to isolate a single point in $E$ using a subspace of appropriate codimension. We will also include the proof (see Appendix~\ref{sec:vvlemma}) due to its simplicity. We state the lemma:

\begin{lemma}[Valiant-Vazirani Lemma.]
\label{lem:vv}
For integers $t$ and $m$ such that $t \leq \frac{1}{100}2^{m/2}$, let $l$ and $c$ be such that $l$ is an integer, and $l = \log_2 ct$ where $2 \leq c < 4$. Given $E \subset  \F_2^m$ such that $|E| = t$, let $a_1,\ldots,a_l$ be uniform among all sets of $l$ linearly independent vectors, and $b_1,\ldots, b_l$ be uniformly and independently chosen elements of $\F_2$. Let 
\[
S \defeq \{x \in E \suchthat \langle x, a_k \rangle = b_k~\forall k \in [l]\}.
\]
Then, for every $e \in E$,

\[
\Pr(S = \{e\}) \geq \frac{1}{7t}.
\]
\end{lemma}

So, if we restrict $100 t\log t$ times, then the probability that some point is never isolated is at most $t\left(1 - \frac{1}{7t} \right)^{100 t \log t} \leq 0.001$. What remains is to ensure that we have \emph{all} the polynomials of the given degree that vanish at that point. This is shown by obtaining this set of polynomials after every affine restriction.

\subsection{Restricting the points to a hyperplane.}
\label{sec:restrict}

Here, we just analyze the case when restricting to $X_m = 0$. Further restrictions are analyzed in exactly the same way.

We have $E_1$, the set of all $z \in \F_2^{m-1}$ such that $(z,0) \in E$.
Let $\hat{E_1}$ be the set of all $z \in \F_2^{m-1}$ such that $(z,1) \in E$.
We also have $V_1$, the space of all $n-1$ variate polynomials that vanish on $E_1$. The following lemma shows that $V_1$ can be found easily from $V$.

\begin{lemma}
\label{lem:polyrestrict}
We have that 
\begin{align*}
V_1 = & ~ \{P(X_1, \ldots, X_{m-1}, 0) \in \F_2[X_1, \ldots, X_{m-1}] \suchthat \\
& ~ P(X_1, \ldots, X_m) \in V\}.
\end{align*}
\end{lemma}

\begin{proof}
Fix a polynomial $P(X_1, \ldots, X_m) \in V$. We first show that the polynomial in $n-1$ variables \linebreak $Q(X_1, \ldots, X_{m-1}) = P(X_1, \ldots, X_{m-1}, 0)$ lies in $V_1$. But this is obvious: since $P(X_1, \ldots, X_m) \in V$, we know that $P(y) = 0$ for all $y \in E$. Thus
for any $z \in E_1$, $P(z,0) = 0$. Thus $Q(z) = 0$.

For the other direction, suppose $Q(X_1, \ldots, X_{m-1}) \in V_1$. We need to show that there exists
some $P(X_1, \ldots, X_m) \in V$ such that $Q(X_1, \ldots, X_{m-1}) = P(X_1, \ldots, X_{m-1}, 0)$. We will show that there is some polynomial $P'(X_1, \ldots, X_{m-1}) \in \F[X_1, \ldots , X_{m-1}]$ of degree  at most $r$ such that \[Q(X_1, \ldots, X_{m-1}) + X_m\cdot P'(X_1, \ldots, X_{m-1}) \in V.\] Towards this, let $(a_M)_{M \in \mathcal{M}_r^{m-1}}$ be indeterminates, and 
let $P'(X_1, \ldots, X_{m-1})$ be given by:

\[
P'(X_1, \ldots, X_{m-1}) = \sum_{M \in \mathcal{M}_r^{m-1}} a_M M(X_1, \ldots X_{m-1}).
\]

We set up a system of linear equations on the $a_M$:

\begin{align*}
Q(z) + P'(z) = 0 & \text{\hspace{10mm} For every $z \in E_1$}
\end{align*}

We claim that there exists a solution $(a_M^*)_{M \in \mathcal{M}_r^{m-1}}$ to this system of equations. This is because $\hat{E_1}$ satisfies property $U_r$. This follows from the fact that for every $e \in \hat{E_1}$, i.e., for any $(e,1) \in E$, the entries of $(e,1)^{\otimes \leq r}$ are the same as the entries in $e^{\otimes \leq r}$ with some entries repeated, so any linear dependency among the columns of $\hat{E_1}^{\otimes \leq r}$ corresponds to a linear dependency in the columns of $E^{\otimes \leq r})$.

Finally, it remains to check that for every such $P'$ (actually, just some $P'$ is enough), we have that \[Q(X_1,\ldots, X_{m-1}) + X_mP'(X_1,\ldots, X_{m-1}) \in V,\] i.e., $Q(X_1,\ldots, X_{m-1}) + X_mP'(X_1,\ldots, X_{m-1})$ vanishes on $E$. But this is obvious: the case when $X_m = 0$ is taken care of by the fact that $Q \in \tilde{S}$, and the case when $X_m = 1$ is handled by the fact that $P'$ is a solution to our system of equations.

\end{proof}

In the above lemma, all the polynomials in $V$ are \emph{reduced}, i.e., have degree in each variable at most one. When we apply an invertible affine transformation on the variables, we have to ensure that all the polynomials are reduced. However, this is again easy, as it suffices to reduce the basis polynomials of the space. Henceforth, for a set of polynomial $P \in \F_2[X_1\ldots, X_n]$, we shall denote $\operatorname{reduce}(P)$ to be polynomial obtained after reducing $P$.

And finally, we present the full algorithm

\begin{algorithm}
\label{alg:main}

\begin{algorithmic}
\\
\Procedure{FindRoots}{$V$}
\State $t \gets \operatorname{codim}(V)$
\State \textbf{Initialize} $E \gets \emptyset$ \Comment{error set}
\For{$100t \log t$ iterations}
\State $M \sim GL(m,\F_2)$, $b \sim \F_2^m$
\For{$P \in V$}
\State $P(X) \gets \operatorname{reduce}( P(MX + b))$ \Comment{affine transformation}
\EndFor
\State  $V_l \gets \{P(X_1,\ldots, X_{n-l}, 0, \ldots, 0) \suchthat P(X_1,\ldots, X_{n}) \in V\}$
\State $e \gets $\Call{FindUniqueRoot}{$V$} 
\If{$e \neq \bot$}
\State $E \gets E \cup \{e\}$
\EndIf
\EndFor
\State \textbf{return} $E$
\EndProcedure
\end{algorithmic}
\end{algorithm}

We do $100 t \log t$ iterations, and in each step the most expensive operation is \textsc{FindUniqueRoot}, which takes time $O(m^{r \omega + 2})$, since it is essentially equivalent to checking if s given vector is in the span of some set of $\leq m^{r}$ vectors . Therefore, the total running time is $O\left(m^{(\omega + 1)r + 4}\right)$

\subsection{A note on derandomization.}
\label{sec:derandomization}

In this section, we show how to run the previous algorithm in a derandomized way. The key tool is that we can count the number of common roots of the space via Equation~\ref{eqn:number} for any instance. So, this suggests a natural approach: we try to restrict variables one by one to $0$ or $1$, and then finding the corresponding space of polynomials by Lemma~\ref{lem:polyrestrict}, only ensuring that the number of common roots after restricting is still nonzero. 

\begin{algorithm}
\label{alg:detmain}

\begin{algorithmic}
\\
\Procedure{DetFindRoots}{$V$}
\State  $V_0 \gets \{P(X_1,\ldots, X_{n-1}, 0) \suchthat P(X_1,\ldots, X_{n}) \in V\}$
\State  $V_1 \gets \{P(X_1,\ldots, X_{n-1}, 1) \suchthat P(X_1,\ldots, X_{n}) \in V\}$
\If{$\operatorname{codim}(V_0) \neq 0$}
\State $E_0 \gets \Call{DetFindRoots}{V_0$}
\Else
\State $E_0 \gets \emptyset$
\EndIf
\If{$\operatorname{codim}(V_1) \neq 0$}
\State $E_1 \gets \Call{DetFindRoots}{V_1$}
\Else
\State $E_1 \gets \emptyset$
\EndIf
\State \textbf{return} $E_0 \cup E_1$
\EndProcedure
\end{algorithmic}
\end{algorithm}

To find the running time, we utilize the following recurrence:
\\
\begin{align*}
T(m,|E|)  \leq &~ T(m-1, |E_0|) + T(m-1, |E_1|) \\
& + \binom{m}{\leq r+1}^{\omega},
\end{align*}
where $|E_0| + |E_1| = |E|$. This gives a running time bound of $O\left(m^{(\omega+1)r + 2}\right)$.

\section{Extension to other small fields.}

The algorithm given above is easily extended to other fields of small order. The reduction of the syndrome decoding problem to finding roots of a space of low degree polynomials, and the isolation lemma can be adapted with almost no change at all. We will only reproduce the result of Section~\ref{sec:restrict}. We do it for $\F_p$ and show that we can recover the whole space of polynomials that vanish on a set of points after one restriction $X_m = 0$. We carry over the notation too. Let $E_1 \defeq \{e \suchthat (e,0) \in E\}$. Let $\hat{E}_1 \defeq E \setminus \{(e,0) \suchthat e \in E_1\} $.

\begin{lemma}
\[
V_1 =  \{P(X_1, \ldots, X_{m-1}, 0) \in \F_p[X_1, \ldots, X_{m-1}] \suchthat P(X_1, \ldots, X_m) \in V\}
\]
\end{lemma}

\begin{proof}
As in the previous case, one direction is obvious. Let $P(X_1,\ldots, X_m) \in V$. For every point $(z,0) \in E$, we have that $P(X_1,\ldots, X_{m-1},0)$ vanishes at $z$.

For the other direction, again similar to the previous case, let $(a_M^{(i)})_{M \in \mathcal{M}_r^{m-1},i \in [p-1]}$ be indeterminates, let the polynomials in the indeterminates $A_1(X_1,\ldots X_{m-1}),\ldots, A_{p-1}(X_1,\ldots X_{m-1})$ be given by:

\begin{align*}
A_1(X_1,\ldots, X_{m-1})  & = \sum_{M \in \mathcal{M}_r^{m-1}}a_M^{(1)}M(X_1, \ldots, X_{m-1}) \\
&\vdots \\
A_{p-1}(X_1,\ldots, X_{m-1}) & = \sum_{M \in \mathcal{M}_{r - p + 1}^{m-1}}a_M^{(p-1)}M(X_1, \ldots, X_{m-1}).
\end{align*}

and consider the system of linear equations:

\begin{align*}
A_1(y^{(1)}) + \cdots + A_{p-1}(y^{(1)}) & = -Q(y^{(1)}) \\
& \text{for}~(y^{(1)}, 1) \in E \\
& \vdots \\
(p-1)A_1(y^{(p-1)}) + (p-1)^{p-1}A_{p-1}(y^{(p-1)}) & = -Q(y^{(p-1)}) \\
& \text{for}~(y^{(p-1)}, p-1) \in E
\end{align*}

Rearranging, we have:

\begin{align*}
A_1(y^{(1)}) + \cdots + A_{p-1}(y^{(1)}) & = -Q(y^{(1)}) \\
& \text{for}~(y^{(1)}, 1) \in E\\
& \vdots \\
A_1(y^{(p-1)}) + (p-1)^{p-2}A_{p-1}(y^{(p-1)}) & = -(p-1)^{-1}Q(y^{(p-1)}) \\
& \text{or}~f(y^{(p-1)}, p-1) \in E
\end{align*}

We claim that a solution exists, and therefore such a polynomial 

\[
Q(X_1,\ldots, X_{m-1}) + \sum_{i \in [p-1]}X_m^iA_{i}(X_1,\ldots, X_{m-1})
\]

vanishes on $E$, and has degree at most $r+1$ and therefore, must belong to $V$. Let us denote, for $i \in [p-1]$, $E_1^{(i)} \defeq \{e \suchthat (e,i) \in E\}$. Writing the coefficients on the L.H.S in matrix form, we get

\[
\left(
\begin{array}{ccc}
  \operatorname{mat}((E_1^{(1)})^{\otimes \leq r})  & \cdots & \operatorname{mat}((E_1^{(1)})^{\otimes \leq r - p + 1})  \\
   \vdots & \ddots & \vdots \\
  \operatorname{mat}((E_1^{(p-1)})^{\otimes \leq r})  & \cdots & (p-1)^{p-2} \operatorname{mat}((E_1^{(p-1)})^{\otimes \leq r - p + 1})  \\
\end{array}
\right)
\]

It is easy to see that the above matrix is constructed by dropping some repeated columns of $\operatorname{mat}((\hat{E}_1)^{\otimes \leq r})$, and therefore, is full rank.

\end{proof}

\section{Discussion and open problems.}

A very nice question of~\cite{ASW:15} is to determine whether Reed-Muller codes achieve capacity for the Binary Symmetric Channel. In the constant distance regime, this would amount to being able to correct $\Theta(m^{r-1})$ random errors in the Reed-Muller code $RM(m,m-r)$. This seems to be closely related to understanding when tensors have high tensor rank, a question of great interest in algebraic complexity theory (eg. see \cite{RAZ:13} and the references therein).

Another interesting problem comes from our second approach to this syndrome decoding problem. Although our algorithm works well over small fields, over large fields, it has a bad dependence on the field size. This mainly comes because when trying to isolate one point using a subspace. It would be interesting to have an algorithm whose running time grows polynomially in $\log p$ instead of $p$, where $p$ is the size of the field. More concretely is there a $\operatorname{poly}(m^r, \log p)$ algorithm for the following problem?

\begin{itemize}

\item \emph{Input:} The space $S$ of all polynomials in $m$ variables of degree at most $r+1$ over $\F_p$ which vanish on an (unknown) set $E$ of points that satisfy property $U_r$.

\item \emph{Output:} The set $E$.
\end{itemize}

\section{Acknowledgements.}

We would like to thank K.P.S.\ Bhaskara Rao, Ramprasad Saptarshi, Amir Shpilka, and Avi Wigderson for helpful discussions. The first author would like to thank Ben Lee Volk for the helpful discussions at an early stage of this work.

\bibliographystyle{alpha}

\bibliography{references}

\appendix
\section{The Valiant Vazirani isolation lemma.}
\label{sec:vvlemma}

In this section, we prove the exact statement that we use to isolate a point using a subspace. This approach for the proof may be found in, for example, \cite{SUDAN}.

\begin{lemma}[Valiant-Vazirani Lemma.]
\label{lem:appendixvv}
For integers $t$ and $m$ such that $t \leq \frac{1}{100}2^{m/2}$, let $l$ and $c$ be such that $l$ is an integer, and $l = \log_2 ct$ where $2 \leq c < 4$. Given $E \subset  \F_2^m$ such that $|E| = t$, let $a_1,\ldots,a_l$ be uniform among all sets of $l$ linearly independent vectors, and $b_1,\ldots, b_l$ be uniformly and independently chosen elements of $\F_2$. Let 
\[
S \defeq \{x \in E \suchthat \langle x, a_k \rangle = b_k~\forall k \in [l]\}.
\]
Then, for every $e \in E$,

\[
\Pr(S = \{e\}) \geq \frac{1}{7t}.
\]
\end{lemma}

\begin{proof}
Assume that $a_1,\ldots, a_l$ are chosen uniformly and independently from $\F_2^m$. let $I$ denote the event $\{a_1,\ldots, a_l \text{ are linearly independent}\}$. We have $\Pr(I) \geq 1 - \frac{ct}{2^{m}}$. We have, for every $i \in [t]$, that $\Pr(\langle e_i, a_k \rangle = b_k) = \frac{1}{2}$. Moreover, we have the pairwise independence property that $\Pr(\langle e_i, a_k \rangle = b_k \land \langle e_j,a_k \rangle = b_k) = \frac{1}{4}$ for $i \neq j$.

With this in mind, let $E = \{e_1,\ldots, e_l\}$, and $\mathcal{E}_i$ denote the event $\{\langle e_i , a_k \rangle = 0 \suchthat k \in [l] \}$. We have that $\Pr(\mathcal{E}_i) = \frac{1}{2^l} = \frac{1}{ct}$, and $\Pr(\mathcal{E}_i \land \mathcal{E}_j) = \frac{1}{4^l} = \frac{1}{c^2t^2}$ for $i \neq j$. We have:

\[
\mathcal{E}_i \subseteq \left(\mathcal{E}_i \cap \left( \bigcap_{j \neq i} \overline{\mathcal{E}_j}\right)\right) \cup \left(\bigcup_{j \neq i}(\mathcal{E}_i \cap \overline{E}_j) \right).
\]

Therefore, by union bound, 

\[
\Pr(\mathcal{E}_i) \leq \Pr \left(\mathcal{E}_i \cap \left(\bigcap_{j \neq i}\overline{\mathcal{E}_j} \right) \right) + \sum_{j \neq i} \Pr(\mathcal{E}_i \cap \mathcal{E}_j),
\]

or

\[
\Pr \left(\mathcal{E}_i \cap \left(\bigcap_{j \neq i}\overline{\mathcal{E}_j} \right) \right) \geq \frac{1}{t}\left(\frac{1}{c} - \frac{1}{c^2} \right) \geq \frac{1}{6t}.
\]

And finally, by the law of total probability,

\begin{align*}
\Pr \left(\mathcal{E}_i \cap \left(\bigcap_{j \neq i}\overline{\mathcal{E}_j} \right)~ \Bigg | ~ I \right)  & \geq \frac{1}{6t} - \frac{ct}{2^m} \\
& \geq \frac{1}{7t}.
\end{align*}

\end{proof}

\end{document}